\documentclass[12pt,a4paper,oneside,lefteqn]{article}
\pdfoutput=1

\usepackage[margin=1.5in]{geometry}

\usepackage[T1]{fontenc}
\usepackage[utf8]{inputenc}
\usepackage[american]{babel}

\usepackage{hyperref}
\usepackage[pdftex]{graphicx}
\usepackage{xspace}
\usepackage{textcomp}
\usepackage{framed}

\usepackage{verbatim}

\usepackage{amsmath}
\usepackage{amssymb}
\usepackage{amsthm}
\usepackage{txfonts}
\usepackage{mathtools}

\newcommand{\Nat}{\mathbb{N}}

\newcommand{\meet}{\wedge}
\newcommand{\join}{\vee}

\newcommand{\bigmeet}{\bigwedge}
\newcommand{\bigjoin}{\bigvee}

\newcommand{\lcm}{\operatorname{\mathrm{lcm}}}

\newcommand{\upsort}[1]{{#1}^{\uparrow}}

\newcommand{\weakupsort}[1]{{#1}^\vartriangle}

\newcommand{\inverse}[1]{{#1}^{-1}}

\newcommand{\eqdef}{\coloneqq}

\theoremstyle{definition}
\newtheorem{definition}{Definition}[section]
\newtheorem{example}[definition]{Example}
\newtheorem{remark}[definition]{Remark}

\theoremstyle{theorem}
\newtheorem{lemma}[definition]{Lemma}

\newtheorem{proposition}[definition]{Proposition}

\title{Sorting in Lattices}
\author{Jens Gerlach\\
        {\small Fraunhofer FOKUS, Berlin, Germany}
       }
\date{}

\begin{document}

\maketitle

\begin{abstract}
In a totally ordered set the notion of sorting a finite sequence
is defined through a suitable permutation of the sequence's indices.
In this paper we prove a simple formula that explicitly describes
how the elements of a sequence are related to those of its sorted counterpart.
As this formula relies only on the minimum and maximum functions
we use it to define the notion of sorting for lattices.
A major difference of sorting in lattices is that it does \emph{not}
guarantee that sequence elements are only rearranged.
However, we can show that other fundamental properties that are associated
with sorting are preserved.
\end{abstract}

\section{Introduction}
Let $(X,\leq)$ be a totally ordered set and $x \in X^n$
a sequence $x_1,\ldots,x_n$ of length~$n$ in~$X$.
There exists for each such sequence a permutation $\varphi$ of $[1,n] = \{1,\ldots,n\}$
such that $x\circ\varphi \in X^n$ is a nondecreasing sequence.
If $x$ is injective, then $\varphi$ is uniquely determined, and vice versa.
However, regardless whether there is exactly one permutation, the
rearrangement $x\circ\varphi$ is uniquely determined and one thus can
refer to it as \emph{the result of nondecreasing sorting~$x$} which we denote as~$\upsort{x}$.
We remark that $\upsort{x}$ is the only nondecreasing sequence of length~$n$
in which each element of $X$ appears as often as in $x$.

Nondecreasing sorting defines a map $x \mapsto \upsort{x}$
from $X^n$ to the subset of nondecreasing sequences.
This map has several interesting properties.
First of all, it is idempotent, that is,
\begin{align}
\label{eq:sorting-idempotent}
     \upsort{\left(\upsort{x}\right)} &= \upsort{x}
\intertext{and thus a projection. This  implies also that the map $x \mapsto \upsort{x}$
  is surjective.
  Secondly, for each permutation $\psi$ of $[1,n]$ we have}
\label{eq:sorting-invariant}
     \upsort{(x\circ\psi)} &= \upsort{x}
\end{align}

In Section~\ref{sec:weak-sorting} we prove Identity~\eqref{eq:sorting-equivalence}
that explicitly describes how the elements of
$\upsort{x}_1,\ldots,\upsort{x}_n$ are related to $x_1, \ldots, x_n$.
This formula only uses the minimum and maximum functions on finite sets.
Based on this observation, we define in Section~\ref{sec:lattice-sorting}
the notion of  \emph{sorting of sequences in a lattice}  through simply 
replacing the minimum\slash maximum operations by the infimum\slash supremum operations.
We also show that sorting in lattices in general not just reorders the elements
of a sequence but really changes them.
However, we also prove that our definition satisfies
other properties that are associated with sorting.

\section{A Formula for Sorting}
\label{sec:weak-sorting}

Let $(X,\leq)$ be a totally ordered set, then each
nonempty finite subset $A$ of $X$ contains a \emph{least}
and a \emph{greatest} element~\cite[R.~6.5]{BourbakiSets}.
We also speak of the minimum and maximum
of $A$ and refer to these special elements as 
$\bigmeet A$ and $\bigjoin A$, respectively.
The following inequalities hold for all $a \in A$
\begin{align}
\label{eq:min-max-bounds}
   \bigmeet A \leq a \leq \bigjoin A
\end{align}

For $A = \{x, y\}$ we use the notation $x \meet y$ and $x \join y$
to denote the minimum and maximum, respectively.

The main results of this paper depend of a particular family
of finite sets.

\begin{definition}
For $k \in [1,n]$ we denote with 
\begin{align*}
   \textstyle{\Nat\binom{n}{k}} &\eqdef \Big\{A \subset [1,n] \bigm| |A| = k \Big\}
\end{align*}
the set of subsets of~$[1,n]$ that contain exactly $k$ elements.
There are $\binom{n}{k}$ such subsets.
\end{definition}

\begin{proposition}
\label{sorting-equivalence}
Let $(x_1,\ldots,x_n)$ be a sequence in a totally ordered set, then
the following identity holds for the elements of the sequence
 $\left(\upsort{x}_1, \ldots,\upsort{x}_n\right)$ 

\begin{align}
\label{eq:sorting-equivalence}
  \upsort{x}_k  &=  \bigmeet_{I \in \Nat\binom{n}{k}} \bigjoin_{i \in I} x_i
\end{align}
\end{proposition}

Before we prove Proposition~\ref{sorting-equivalence} we
introduce an abbreviation for the right hand side
of Identity~\eqref{eq:sorting-equivalence}.
For a sequence~$x$ of length~$n$ we define

\begin{align}
\label{eq:weakupsort}
   \weakupsort{x}_k &\eqdef \bigmeet_{I \in \Nat\binom{n}{k}} \bigjoin_{i \in I} x_i
     && \text{for $1 \leq k \leq n$}
\end{align}

With this notation Proposition~\ref{sorting-equivalence} reads $\upsort{x} = \weakupsort{x}$.

Here are some simple observations about the elements of~$\weakupsort{x}$.

\begin{itemize}
\item
Since $(X,\leq)$ is a total order, we know that each element of~$\weakupsort{x}$
belongs to~$x$.

\item
In particular, we see that $\weakupsort{x}_1$ is the least element and
$\weakupsort{x}_n$ the greatest element of~$x$, respectively.
\end{itemize}

The following lemma states that $\weakupsort{x}$ is a nondecreasing sequence.
\begin{lemma}
\label{weakupsort-nondecreasing}
If $x$ is a sequence of length~$n$ in a totally ordered set $(X,\leq)$, 
then~$\weakupsort{x}$ is a nondecreasing sequence.
\end{lemma}

\begin{proof}
Let $1 \leq k < n$ and  $I$ be an arbitrary subset of~$[1,n]$ with $k+1$ elements.
If $J$ is a subset of~$I$ with $k$ elements, then we have
\begin{align*}
   \weakupsort{x}_k 
     = \bigmeet_{L\in\Nat\binom{n}{k}} \bigjoin_{l \in L} x_l 
     &\leq \bigjoin_{j \in J} x_j
           && \text{by Inequality \eqref{eq:min-max-bounds}}\\
     &\leq \bigjoin_{i \in I} x_i 
           && \text{by $J \subset I$}\\
\intertext{Since $I$ is an arbitrary set of~$k+1$ elements we obtain from here}
   \weakupsort{x}_k
      \leq \bigmeet_{I\in\Nat\binom{n}{k+1}} \bigjoin_{i \in I} x_i
      &= \weakupsort{x}_{k+1}
\end{align*}
\end{proof}

\begin{remark}
\label{remark-weakupsort}
Note that in the proof of Lemma~\ref{weakupsort-nondecreasing}
we have only used the fact that the minimum of a set is 
a \emph{lower bound} for all elements of that set.
\end{remark}

\begin{proof}[Proof of Proposition~\ref{sorting-equivalence}]
We will show that for each $k$ with $k \in [1,n]$ both
\begin{align*}
   \weakupsort{x}_k \leq \upsort{x}_k &&\text{and}&& \upsort{x}_k \leq \weakupsort{x}_k
\end{align*}
hold.

Let $\varphi$ be a permutation of~$[1,n]$ with 
\begin{align}
\label{eq:sorting-permutation}
\upsort{x} &= x\circ\varphi
\intertext{and let $J \subset [1,n]$ be the subset for which} 
\label{eq:permutation-of-interval}
J &= \varphi\left([1,k]\right)
\end{align}
holds.

From the fact that $J$ contains exactly~$k$ elements and that~$\upsort{x}$
is nondecreasing we conclude
\begin{align*}
   \weakupsort{x}_k =
      \bigmeet_{I \in \Nat\binom{n}{k}} \bigjoin_{i \in I} x_i 
      &\leq \bigjoin_{j \in J} x_j
           && \text{by Inequality \eqref{eq:min-max-bounds}}\\
      &= \bigjoin_{j \in J} \upsort{x}\left(\inverse{\varphi}(j)\right)
           && \text{by Identity \eqref{eq:sorting-permutation}}\\
      &= \bigjoin_{i \in [1,k]} \upsort{x}_i 
           && \text{by Identity \eqref{eq:permutation-of-interval}}\\
\intertext{According to Lemma~\ref{weakupsort-nondecreasing}, the sequence
 $\upsort{x}$ is nondecreasing and we obtain}
      &= \upsort{x}_k
\end{align*}
which finishes the first part of the proof.

Conversely, we conclude from the definition of $\weakupsort{x}$
and the fact that $(X,\leq)$ is a total order that there exists a subset $B$ of~$[1,n]$
with exactly $k$ elements such that
\begin{align*}
    \weakupsort{x}_k = \bigmeet_{I \in \Nat\binom{n}{k}} \bigjoin_{i \in I} x_i
        &= \bigjoin_{i \in B} x_i \\
        &= \bigjoin_{i \in B} \upsort{x}\left(\inverse{\varphi}(i)\right) 
           && \text{by Identity \eqref{eq:sorting-permutation}}\\
        &= \bigjoin_{j \in \inverse{\varphi}(B)} \upsort{x}_j
\end{align*}
holds.
Since $\upsort{x}$ is nondecreasing we have
\begin{align}
\nonumber
   \bigjoin_{j \in \inverse{\varphi}(B)} \upsort{x}_j &= \upsort{x}_m
         && \text{where} &&
   m = \bigjoin(\inverse{\varphi}(B))
\intertext{is the greatest element of $\inverse{\varphi}(B)$.
Thus, we have}
\label{eq:almost-there}
  \weakupsort{x}_k &= \upsort{x}_m
\end{align}

However, since $\bigjoin(\inverse{\varphi}(B))$ is a subset
of~$[1,n]$ that has exactly $k$ elements we have
\begin{align*}
  k &\leq m
\intertext{and since $\upsort{x}$ is nondecreasing}
  \upsort{x}_k &\leq \upsort{x}_m
\intertext{From this and Identity~\eqref{eq:almost-there}  we conclude}
   \upsort{x}_k &\leq \weakupsort{x}_k
\end{align*}
which completes the proof.
\end{proof}

\section{Sorting in Lattices}
\label{sec:lattice-sorting}

Let $(X,\leq)$ be a \emph{partial order} that is also a \emph{lattice} $(X,\meet,\join)$~\cite{Graetzer2003}, that is,
for each $x, y \in X$ there exists the  {infimum}  $x \meet y$
and the supremum $x \join y$.
These operations are commutative and associative.
Moreover, they satisfy the so-called
absorption properties for all $x, y \in X$
\begin{align*}
x \join (x\meet y)  &= x \\
x \meet (x \join y) &= x
\end{align*}

In a lattice, the infimum and supremum exist for every finite subset~$A$~\cite[p.~4]{Graetzer2003}
and are denoted by $\bigmeet A$ and $\bigjoin A$, respectively.
Note, however, that for a finite subset $A$ in a general lattice neither the infimum nor the supremum
necessarily belong to~$A$.
If $(X, \leq)$ is a \emph{total order}, then $\bigmeet$ and $\bigjoin$ are the
minimum and maximum functions.
This means, our notation is consistent with that of Section~\ref{sec:weak-sorting}.

An essential observation is that for a sequence $x$ of length~$n$ in a lattice the
value
\begin{align*}
   \bigmeet_{I \in \Nat\binom{n}{k}} \bigjoin_{i \in I} x_i 
\end{align*}
is well-defined for $k \in [1,n]$.
This motivates the following definition.

\begin{definition}
If $x$ is a sequence of length~$n$ in a lattice $(X,\meet,\join)$, then
we refer to $\weakupsort{x}$ as defined by Identity~\eqref{eq:weakupsort}
as \emph{$x$ nondecreasingly sorted with respect to the lattice $(X,\meet,\join)$}.
\end{definition}

The following lemma states that for $\weakupsort{x}$ is indeed a 
nondecreasing sequence with respect to the partial order $(X,\leq)$ of the lattice.

\begin{lemma}
\label{lattice-sorting-nondecreasing}
If $x$ is a finite sequence in a lattice $(X,\meet,\join)$ with associated partial 
order $(X,\leq)$, then Identity~\eqref{eq:weakupsort} defines a nondecreasing 
sequence $\weakupsort{x}$.
\end{lemma}

\begin{proof}
In order to prove this lemma we can proceed exactly as in the proof
of~Lemma~\ref{weakupsort-nondecreasing} where $(X,\leq)$ is a total order.
As noted in Remark~\ref{remark-weakupsort}, we have used only
the fact that $\bigmeet A$ is a {lower bound}
of~$A$ which by definition also holds for lattices.
\end{proof}

\pagebreak

A simple consequence of Identity~\eqref{eq:weakupsort} and Lemma~\ref{lattice-sorting-nondecreasing}
is that sorting in lattices respects lower and upper bounds of the original sequence.

\begin{lemma}
\label{lattice-sorting-bounds}
Let $x = (x_1,\ldots,x_n)$ be a finite sequence in a lattice $(X,\meet,\join)$ with associated partial
order $(X,\leq)$.
If for $1 \leq i \leq n$ holds
\begin{align*}
    a &\leq x_i \leq b,
\intertext{then}
    a &\leq \weakupsort{x}_i \leq b
\end{align*}
holds for $1 \leq i \leq n$ as well.
\end{lemma}

\begin{proof}
From Identity~\eqref{eq:weakupsort} (see also Identity~\eqref{eq:weakupsort-explicit})
follows that~$\weakupsort{x}_n$ is the supremum of the elements
$x_1,\ldots,x_n$. Thus, we have $\weakupsort{x}_n \leq b$.
Lemma~\ref{lattice-sorting-nondecreasing} ensures 
that~$\weakupsort{x}_n$ is the largest element of~$\weakupsort{x}$.
Thus we have $\weakupsort{x}_i \leq b$ for $1 \leq i \leq n$.
The case for the lower bound~$a$ is treated analogously.
\end{proof}

\subsection{Examples}
\label{sec:lattice-sorting-examples}

When applying Identity~\eqref{eq:weakupsort} it is sometimes convenient to
use a slightly  more explicit way to write the elements of $\weakupsort{x}$.

%\begin{equation}
%\begin{aligned}
\begin{align}
\nonumber
   \weakupsort{x}_1 &= x_1 \meet \ldots \meet x_n \\
\nonumber
   \weakupsort{x}_2 &= \bigmeet_{1 \leq i < j \leq n} x_i \join x_j\\
\nonumber
   &\,\vdotswithin{=}  \\
\label{eq:weakupsort-explicit}
   \weakupsort{x}_k &= \bigmeet_{1 \leq i_1 < \ldots < i_k \leq n} x_{i_1} \join \ldots \join x_{i_k}\\
\nonumber
   &\,\vdotswithin{=} \\
\nonumber
   \weakupsort{x}_n &= x_1 \join \ldots \join x_n.
\end{align}
%\end{aligned}
%\end{equation}

\begin{example}
Consider the finite set $X = \{x, y, z\}$.
Figure~\ref{fig:lattice-set} shows the lattice of all subsets of $X$.

\begin{figure}[hbt]
\centering
\includegraphics[width=0.50\linewidth]{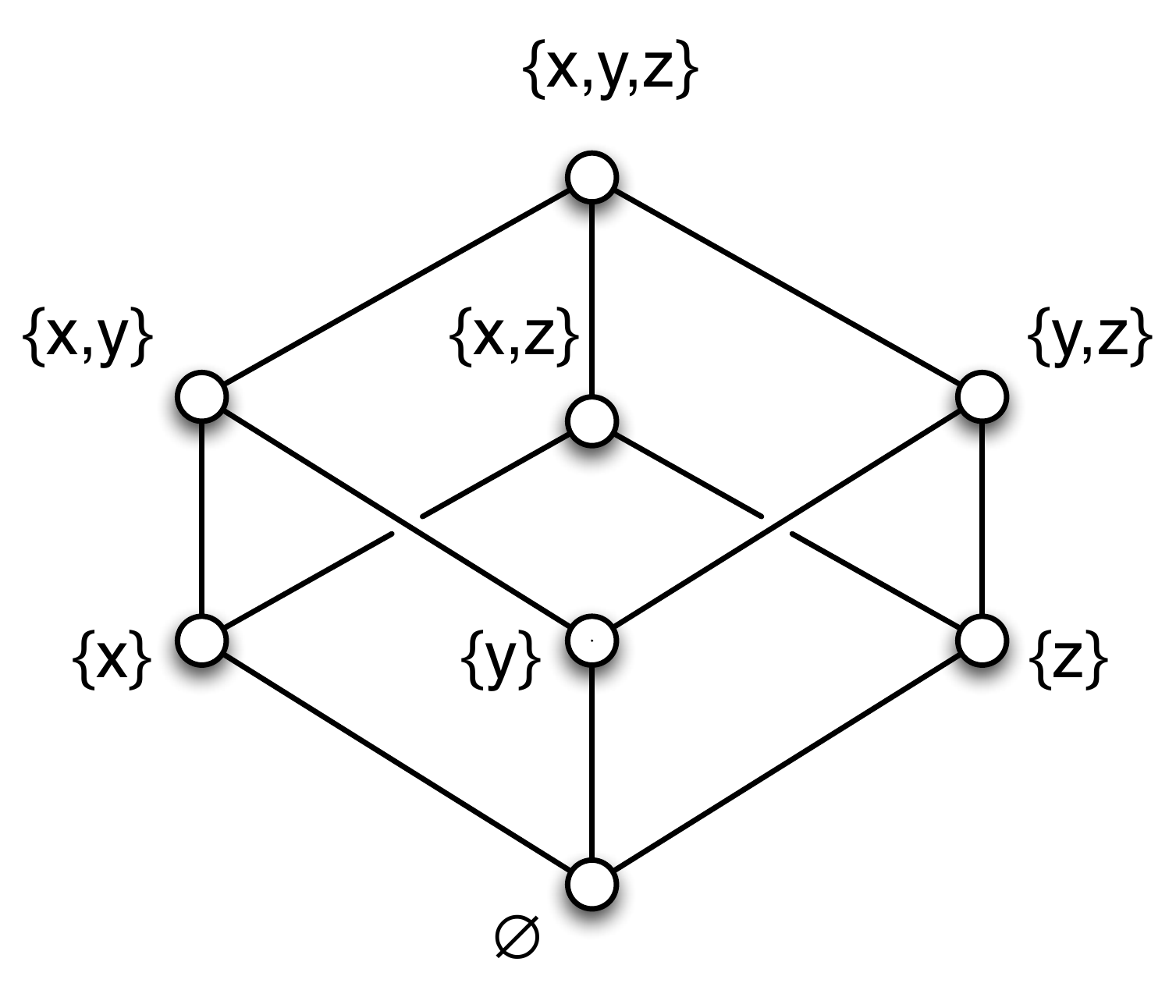}
\caption{\label{fig:lattice-set} The lattice of $\{x, y , z\}$}
\end{figure}

Let $x$ be the sequence
\begin{align*}
   a &= \bigl(\{x\}, \{y\}, \{z\}\bigr)
\intertext{then}
   \weakupsort{a} &= \bigl(\emptyset, \emptyset, X\bigr)
\end{align*}
Thus, $\weakupsort{a}$ is a nondecreasing sequence that consists of elements that are completely
different from those of~$a$.
\end{example}

\begin{example}
Let us consider now the lattice $(\Nat,\gcd,\lcm)$
where $\gcd(x, y)$ and $\lcm(x, y)$ denote the \emph{greatest common divisor}
and \emph{least common multiple} of $x$ and $y$, respectively.
The associated partial order of this lattices is defined by divisibility of natural numbers.
Table~\ref{tbl:sequences} shows some examples of our definition of sorting for
different sequences in $(\Nat,\gcd,\lcm)$.
Again we see that sorting in a lattice may change the elements in a sequence.

\begin{table}[hbt]
\begin{center}
\begin{tabular}{|c|c|}
\hline
$x$ & $\weakupsort{x}$  \\
\hline
$(1)$  &  $(1)$  \\
$(1,2)$  &  $(1,2)$  \\
$(1,2,3)$  &  $(1,1,6)$  \\
$(1,2,3,4)$  &  $(1,1,2,12)$  \\
$(1,2,3,4,5)$  &  $(1,1,1,2,60)$  \\
$(1,2,3,4,5,6)$  &  $(1,1,1,2,6,60)$  \\
$(1,2,3,4,5,6,7)$  &  $(1,1,1,1,2,6,420)$  \\
$(1,2,3,4,5,6,7,8)$  &  $(1,1,1,1,2,2,12,840)$  \\
\hline
\end{tabular}
\caption{\label{tbl:sequences}Some examples of sorting in $(\Nat, \gcd, \lcm)$}.
\end{center}
\end{table}

\end{example}

\subsection{Elementary Properties}
\label{sec:lattice-sorting-properties}

In this section we prove that some well-known properties of sorting in 
a totally ordered set also hold for our definition of sorting in lattices.

The following lemma restates the idempotence of sorting in a totally ordered set,
expressed by Identity~\eqref{eq:sorting-idempotent},
for the case of lattices.

\begin{lemma}
If $x$ is a finite sequence in a lattice $(X,\meet,\join)$, then 
\begin{align*}
   \weakupsort{\left(\weakupsort{x}\right)} = \weakupsort{x}
\end{align*}
\end{lemma}

\begin{proof}
We know from Lemma~\ref{lattice-sorting-nondecreasing} that $\weakupsort{x}$ is
a nondecreasing sequence in the partial order $(X,\leq)$.
Thus, the relation~$\leq$ is a \emph{total order} on the set
\begin{align*}
  \left\{\weakupsort{x}_1,\ldots,\weakupsort{x}_n\right\} \subset X
\end{align*}
In other words we can sort $\weakupsort{x}$ in the classical sense.
From this follows
\begin{align*}
   \weakupsort{x} &= \upsort{\left(\weakupsort{x}\right)}\\
                  &= \weakupsort{\left(\weakupsort{x}\right)} 
                       && \text{by Identity~\eqref{eq:sorting-equivalence}}
\end{align*}
\end{proof}

We now restate the invariance of sorting under permutations---see Identity~\eqref{eq:sorting-invariant}.

\begin{lemma}
If $x$ is a finite sequence in a lattice $(X,\meet,\join)$ and $\psi$ a permutation
of~$[1,n]$, then 
\begin{align*}
   \weakupsort{(x\circ\psi)} = \weakupsort{x}
\end{align*}
holds.
\end{lemma}

\begin{proof}
We have for $k \in [1,n]$
\begin{align*}
  \weakupsort{(x\circ\psi)}_k &= \bigmeet_{A \in \Nat\binom{n}{k}} \bigjoin_{i \in A} x(\psi(i)) \\
     &= \bigmeet_{A \in \Nat\binom{n}{k}} \bigjoin_{j\in \psi(A)} x(j) \\
     &= \bigmeet_{B \in \psi\left(\Nat\binom{n}{k}\right)} \bigjoin_{j\in B} x(j)
\intertext{Because $\psi$ is a permutation of $[1,n]$ we find that
  $\psi\left(\Nat\binom{n}{k}\right) = \Nat\binom{n}{k}$ and conclude}
     &= \bigmeet_{B \in \Nat\binom{n}{k}} \bigjoin_{j\in B} x(j)  \\
     &= \weakupsort{x} 
\end{align*}
\end{proof}

\section{Conclusion}

Proposition~\ref{sorting-equivalence} states through
Identity~\eqref{eq:sorting-equivalence} an explicit relationship
between the elements of a
finite sequence in a totally ordered sets to its sorted counterpart.

The author does not suggest that Identity~\eqref{eq:sorting-equivalence}
is an efficient algorithm for sorting.
Since there are~$2^n$ subsets of~$[1,n]$,
a straightforward implementation leads to an algorithm of exponential complexity.
Note that the proven identity bears some similarity to the Binomial Theorem of elementary algebra.
The main benefit of that proposition is \emph{not} to efficiently compute $(a+b)^n$
but to serve as a means for useful transformations in proofs and algorithms.

Using Identity~\eqref{eq:sorting-equivalence} we are able to
define the notion of sorting finite sequences in lattices.
Compared to sorting in a totally ordered set, sorting in lattices is a more
\emph{invasive} procedure because, in general, it  changes sequence elements.
However, the definition maintains other elementary properties that are associated with sorting.

\section{Acknowledgment}
The author would like to express his gratitude for the valuable suggestions
of his colleagues Jochen Burghardt and Hans Werner Pohl.

%\clearpage

%\tableofcontents
%\listoffixmes


\begin{thebibliography}{9}

\bibitem{BourbakiSets}
Bourbaki, N.,
Elements of Mathematics, Theory of Sets,
Addison-Wesley, Reading, MA,
1968.

\bibitem{Graetzer2003}
Gr{\"a}tzer, G.,
General Lattice Theory,
Birkh{\"a}user, Basel,
2003.

\end{thebibliography}
\end{document}